\documentclass[twocolumn,final]{IEEEtran}
\usepackage{amsmath,amsfonts}
\usepackage{amsthm}
\usepackage{algorithmic}
\usepackage{algorithm}
\usepackage{array}
\usepackage[caption=false,font=normalsize,labelfont=sf,textfont=sf]{subfig}
\usepackage{textcomp}
\usepackage{stfloats}
\usepackage{url}
\usepackage{verbatim}
\usepackage{graphicx}
\usepackage{cite}
\usepackage[mathscr]{eucal}
\usepackage{cases}
\usepackage{amssymb, amscd, amsthm}

\usepackage{makecell}
\newtheorem{lemma}{\textbf{Lemma}}
\newtheorem{theorem}{\textbf{Theorem}}
\newtheorem{definition}{\textbf{Definition}}
\newtheorem{remark}{\textbf{Remark}}
\newtheorem{corollary}{\textbf{Corollary}}

\newtheorem{Conjecture}{\textbf{Conjecture}}

\newcommand{\Tr}{{{\rm Tr}}}

\begin{document}

\title{Walsh Spectrum and Boomerang Properties of Locally-APN Niho Functions \thanks{The research of Y. Cui and J. Luo was supported by National Natural Science Foundation of China (Nos.12441102, 12171191) and the Fundamental Research Funds for the Central Universities (No.CCNU25JCPT031). The research of C. Xiang was supported by the Basic and Applied Basic Research Foundation of Guangdong Province of China under Grant number 2026A1515011223 and the National Natural Science Foundation of China under grant number 12171162.}
}

\author{ Yuehui Cui, Jinquan Luo and Can Xiang
\thanks{Y. Cui and J. Luo are with School of Mathematics
and Statistics \& Hubei Key Laboratory of Mathematical Sciences, Central China Normal University, Wuhan China 430079.(E-mail: hfcyh1@163.com; luojinquan@ccnu.edu.cn)}
\thanks{C. Xiang is with College of Mathematics and Informatics, South China Agricultural University, Guangzhou, Guangdong 510642, China.(E-mail: cxiangcxiang@hotmail.com)}}

\markboth{Journal of \LaTeX\ Class Files,~Vol.~, No.~, ~2026}%
{Shell \MakeLowercase{\textit{et al.}}: A Sample Article Using IEEEtran.cls for IEEE Journals}


\maketitle

\begin{abstract}

Recently, the Walsh spectrum and boomerang properties of special power functions have aroused widespread research interest, owing to their important applications in cryptography and information security. In particular, locally-APN functions may offer superior resistance against differential cryptanalysis compared to other functions of equivalent differential uniformity. Up till now only a small number of locally-APN functions have been studied. In this paper, we show that a Niho type power function is locally-APN if and only if its Walsh spectrum takes four values in \(\{-p^m, 0, p^m, 2p^m\}\). Equivalently, the associated cyclic codes have four nonzero weights: $p^{m-1}(p-1)(p^m + k)$ for $k = 0, 1, -1, -2$.  Moreover, we also study properties of Niho type locally-APN power functions, including their differential spectrum, Walsh spectrum, Feistel Boomerang Connectivity Table(FBCT for short) and second-order zero differential spectra.

\end{abstract}

\begin{IEEEkeywords}
Locally-APN function, Niho exponent, Walsh spectrum,  Cyclic codes, Boomerang properties.
\end{IEEEkeywords}

\section{Introduction}
\quad Let $\mathbb{F}_{p^n}$ be the finite field with $p^n$ elements, where $p$ is a prime and $n$ is a positive integer. Let $\mathbb{F}_{p^n}^*$ denote the multiplicative group of $\mathbb{F}_{p^n}$. For a function $F:\mathbb{F}_{p^n}\rightarrow\mathbb{F}_{p^n}$, the derivative function of $F$ at $a \in \mathbb{F}_{p^n}$ is  defined by
\[
D_aF(x)=F(x+a)-F(x), \mbox{ } \forall x \in \mathbb{F}_{p^n}.
\] For any $a,b \in \mathbb{F}_{p^n}$, let
\[
  \triangle_F(a,b)=\#\left\{x \in \mathbb{F}_{p^n} \mid D_aF(x)=b \right\}.
\]
The differential uniformity of $F$ is defined as
\[
  \triangle_F=\max\left\{\triangle_F(a,b) \mid a \in \mathbb{F}_{p^n}^*, b \in \mathbb{F}_{p^n}  \right\}.
\]

Let $F(x)=x^d$ be a power function from $\mathbb{F}_{p^{2m}}$ to $\mathbb{F}_{p^{2m}}$, where $d$ is a positive integer. A positive integer $d$ is called a Niho exponent over $\mathbb{F}_{p^{2m}}$ if
\begin{equation}\label{Niho d}
  d \equiv 1 \,\,({\rm mod}\,\,p^m-1).
\end{equation}
It is well known that the power functions $F(x)$ are often served as Substitution box (S-box) candidates due to their algebraic simplicity and efficiency in hardware implementations. Their inherent algebraic structure also simplifies the analysis of differential properties. For such functions, the equation
\begin{equation*}
(x + a)^d - x^d = b \Leftrightarrow a^d \left( \left( \frac{x}{a} + 1 \right)^d - \left( \frac{x}{a} \right)^d \right) = b
\end{equation*}
reveals the scaling relationship $\triangle_F(a, b) = \triangle_F\left(1, \frac{b}{a^d}\right)$ for all $a \in \mathbb{F}_{p^n}^*$ and $b \in \mathbb{F}_{p^n}$. Consequently, the differential behavior of $F(x)$ is entirely governed by the entries $\triangle_F(1, b)$ as $b$ varies over $\mathbb{F}_{p^n}$. This leads to the following formal definition.

\begin{definition}\rm \cite{firstdef}
  Let $F(x)=x^d$ be a power function from $\mathbb{F}_{p^n}$ to $\mathbb{F}_{p^n}$ with differential uniformity $\triangle_F$. Denote
\begin{equation*}
  \omega_i=\#\left\{b \in \mathbb{F}_{p^n} \mid \triangle_F(1,b)=i \right\},
\end{equation*}
where $0\leq i\leq\triangle_F$. The differential spectrum of $F(x)$ is defined as
\begin{equation*}
  \mathbb{DS}=\left\{\omega_i \mid 0\leq i\leq\triangle_F \mbox{ and } \omega_i>0 \right\}.
\end{equation*}
\end{definition}
\medskip

A function $F$ is called  perfect nonlinear(PN for short) and almost perfect nonlinear(APN for short) if $\triangle_F = 1$ and $\triangle_F = 2$, respectively. As an extension of the APN property, locally-APN power functions were introduced by Blondeau, Canteaut, and Charpin \cite{bl2t}, who noted their potential as strong candidates for resisting differential attacks.
A power function $F$ is called locally-APN if
$$\max\left\{ \triangle_F(1,b) \ \big| \ b \in \mathbb{F}_{p^n} \setminus \mathbb{F}_p \right\} = 2.$$

Locally-APN power functions have recently received a lot of attention, as these functions with some special properties have very important applications in cryptography and information security. For example, Blondeau and Nyberg \cite{exam4} used a cryptographic toy example to show that a locally-APN S-box can yield lower differential probabilities than differentially 4-uniform S-boxes. This suggests that locally-APN functions may offer superior resistance against differential cryptanalysis compared to other functions of equivalent differential uniformity. However, there are only a handful of locally-APN power functions in the literature(see Table \ref{table1}).  Thus, constructing more locally-APN power functions is also very meaningful.
For deeper insights into PN, APN and locally-APN functions, we refer interested readers  to \cite{APN1,bl2t,Blondeau.Perrin,Browning,APN2,Budaghyan,Berger,DCC.hu,L-apn.bin} and references therein.
\begin{table}[ht]
\caption{Locally-APN $F(x)=x^d$ over $\mathbb{F}_{p^n}$ with known differential spectrum}\label{table1}
\centering
\begin{tabular}{|c|c|c|c|}
\hline
    $d$       & Conditions  & $\triangle_F$ & Ref.   \\
    [0.5ex]
    \hline
     $2^n-2$ & $p=2$, $n$ even & 4 &   \cite{firstdef,edd.cc} \\  \hline
    $2^{m+1}-1$  & $p=2$, $n=2m$ & $2^m$ & \cite{bl2t,depth} \\   \hline
    $\frac{2^m-1}{2^k+1}+1$  &  $p=2$, $n=2m$, $(k,m)=1$  & $2^m$ & \cite{TIT.he} \\ \hline
       $\frac{3^n+5}{2}$  &  $p=3$, $n=2m$ & $\frac{3^n+3}{4}$ & \cite{du} \\   \hline
            $2p^{m}-1$  &  $p$ prime, $n=2m$ & $p^m$ & \cite{yan.niho} \\ \hline
    $k(p^m-1)$ &  \makecell{$p$ prime, $n=2m$, \\ $(k,p^m+1)=1$ } & $p^m-2$ & \cite{ DCC.hu}  \\ \hline
       \end{tabular}
\end{table}
It is well known that it is very important to choose cryptographic functions with good properties for resisting security attacks in cryptosystem. A further prevalent technique for attacking symmetric cryptography is linear attack. The nonlinearity of an S-box, a key indicator of its robustness against such linear attacks, is linked to its Walsh spectrum. In addition to differential and linear attacks, the boomerang attack is another important cryptanalysis technique introduced by Wagner\cite{wanger.boom} against block ciphers that involve S-boxes. It can be regarded as an extension of the differential attack. In 2018, Cid et al.\cite{BCT} presented a new tool known as the Boomerang Connectivity Table(BCT for short) to evaluate the resistance of an S-box to boomerang attacks. To address ciphers based on the Feistel Network structure, Boukerrou et al.\cite{fbtdef} expanded this approach to Feistel ciphers, where the S-boxes might not be bijective, and introduced the Feistel Boomerang Connectivity Table. In fact, the coefficients of the FBCT are closely connected to the second-order zero differential spectra of functions over finite fields of even characteristic. The readers can refer to \cite{yue.24,man2} for more information on second-order zero differential spectra. Thus, it is necessary to analyze and study the properties (for example, differential spectrum, walsh spectrum, boomerang spectrum and the second-order zero differential spectra) of some special functions.

Motivated by the above facts and analysis, we will consider Niho-type power functions and analyze their properties in this paper. The objectives of this paper are described in the following three paragraphs.

As special power functions, Niho type power functions have been extensively studied in areas such as cross correlation, bent functions, linear codes, and permutation polynomials. For further details, the reader can refer to \cite{appcode1,survey,Li.shu,Niho1972,xiong,xia1,xia2,TIT.he} as well as the comprehensive book \cite{bookcc} and the references therein.
It is well known that differential uniformity serves as a fundamental metric in cryptography, measuring the resistance of a S-box against differential cryptanalysis \cite{DES1}. However, up to now, only a limited number of classes of Niho type power functions with known differential uniformity. In this paper, we will completely determine the differential uniformity of Niho type power functions. Then, based on the favorable cryptographic properties of locally-APN functions, we will determine the differential spectrum of Niho type locally-APN functions.

A long-standing challenge in the study of the Walsh transform is the characterization of cryptographic functions whose Walsh spectrum assume only a limited number of distinct values, along with the precise determination of their value distributions. This area has attracted considerable research interest. Specially, Helleseth, Lahtonen, and Rosendahl \cite{Helleseth.four} demonstrated that the Walsh transform of Niho  type power functions takes at least four distinct values. In this paper, we will show that these four-valued Niho exponents exactly give all locally-APN Niho type power functions.  It is conjectured that a class of Niho type exponents may give all four-valued Walsh spectrum in \cite{via}. If this conjecture holds, then according to this paper, this class of Niho exponents would give all Niho type locally-APN power functions.

It is well known that a function \( F(x) \) is APN if and only if \( {\rm FBCT}_F(a,b) \) equals 0 for all pairs \( a,b\in\mathbb{F}_{2^n} \) satisfying \( ab(a+b) \neq 0 \) (see \cite{fbtdef,Carlet.C}). However, the behavior of FBCT for locally-APN functions is not well understood, and their second-order zero differential spectra remains largely unexplored.  In this paper, we thoroughly investigate and compute all components of the second-order zero differential spectra for Niho type locally-APN functions. Moreover, we give an exact counting formula on the number of pairs \( (a,b) \in \mathbb{F}_{p^n}^2 \) that occur in the spectra.

The rest of this paper is organized as follows. In Section \ref{Preliminaries}, we introduce some notations and auxiliary tools. In Section \ref{sec niho}, we determine the differential uniformity of Niho type power function $F$. When $F$ is locally-APN, we conduct an investigation of its differential spectrum, Walsh spectrum, associated cyclic codes and second-order zero differential spectra. The conclusive remarks are given in Section \ref{conclusion}.

\section{Preliminaries}\label{Preliminaries}
In this section, we briefly recall some definitions and results which will be used later in this paper. We begin this section
by fixing some notations throughout this paper unless otherwise stated.

\begin{tabular}{lll}
$\mathbb{F}_{p^n}$ & finite field of $p^n$ elements &\\
$\mathbb{F}_{p^n}^*$ & multiplicative group of $\mathbb{F}_{p^n}$ &\\
$\mathbb{F}_{p^n}^{\sharp}$ &  $\mathbb{F}_{p^n}\setminus \{0,-1\}$ &\\
$\psi$ & a primitive element of $\mathbb{F}_{p^n}$ &\\
$\Tr^n_1$ & trace function from $\mathbb{F}_{p^n}$ to $\mathbb{F}_p$ &\\
$\zeta_p$ &  a primitive complex $p$-th root of unity $e^{2\pi \sqrt{-1}/p}$ &\\
$\mu_e$ &  $\{x\in \mathbb{F}_{p^n}\mid x^e=1\}$ &\\
$\overline{x}$ &  $\overline{x}=x^{p^m}$ for $x\in\mathbb{F}_{p^{2m}}$ &\\
$\operatorname{ind}_{\psi}(x)$ & $\operatorname{ind}_{\psi}(x)=t$ for $0 \le t < p^n - 1$, $x\in\mathbb{F}_{p^n}^*$, $x = \psi^t$ &\\
$C_{i,j}$ & $\left\{x\in\mathbb{F}_{p^{2m}}^{\sharp}: \operatorname{ind}_{\psi}(x+1)\equiv i\pmod{p^m+1},\right. $ \\
                                 & $\left.~~~~~~\ \operatorname{ind}_{\psi}(x)\equiv j\pmod{p^m+1} \right\}$.
\end{tabular}
\medskip

Next we introduce two definitions which were described in \cite{fbtdef} and \cite{yue.24}, respectively.
\begin{definition} \rm\label{FBCT.def} \cite{fbtdef}
Let $F:\mathbb{F}_{2^n}\rightarrow\mathbb{F}_{2^n}$. The Feistel Boomerang Connectivity Table(FBCT for short) is an $\mathbb{F}_{2^n}\times\mathbb{F}_{2^n}$ table defined for $(a,b)\in\mathbb{F}_{2^n}^2$ by
\begin{equation*}
 \begin{aligned}
& {\rm FBCT}_F(a,b)=\\
&\makecell{\#\{x \in \mathbb{F}_{2^n}:  F(x+a+b)+F(x+a)+ \\~~~~~~~~~~~~~~~~~~~~~~ F(x+b)+F(x)=0 \}}.
   \end{aligned}
\end{equation*}
It is easy to see that ${\rm FBCT}_F(a,b)=2^n$ when $ab(a+b)=0$. Thus, the Feistel boomerang uniformity of $F(x)$ is defined by
$$
\beta_c(F)=\max _{(a, b) \in \mathbb{F}_{2^n}^2, ab(a+b)\ne 0} {\rm FBCT}_{F}(a, b).
$$
\end{definition}
\medskip

\begin{definition} \rm \cite{yue.24}
Let $F:\mathbb{F}_{p^n}\rightarrow\mathbb{F}_{p^n}$ and $a,b \in \mathbb{F}_{p^n}$, the second-order zero differential spectra of $F$ with respect to $a,b$ is defined as
\begin{equation*}
 \begin{aligned}
&\nabla_F(a,b)=\\
&\makecell{\# \{x \in \mathbb{F}_{p^n}:  F(x+a+b)-F(x+a)-\\~~~~~~~~~~~~~~~~~~~~~~ 
F(x+b)+F(x)=0 \} }.
   \end{aligned}
\end{equation*}
The second-order zero differential uniformity of $F$ is defined by $\nabla_F=\max\left\{\nabla_{F}(a,b) \mid a\neq b, a,b \in \mathbb{F}_{2^n}^*\right\}$ for $p=2$ and $\nabla_F=\max\left\{\nabla_{F}(a,b) \mid a,b \in \mathbb{F}_{p^n}^*\right\}$ for $p>2$.
\end{definition}
\medskip

We will need the results in the following three lemmas, which were documented in \cite{Li.shu}, \cite{luo111} and \cite{importentide}, respectively.

\begin{lemma}\rm\label{niho the}
\cite[Lemma 1]{Li.shu}. Let $d_1=g_1(p^m-1)+1$ and $d_2=g_2(p^m-1)+1$ be two Niho exponents respect to $\mathbb{F}_{p^{2m}}$. Then for $u,v\in\mathbb{F}_{p^{2m}}$, we have
\begin{equation*}
  \sum\limits_{x\in\mathbb{F}_{p^{2m}}}{\zeta_p}^{\Tr_1^{2m}(ux^{d_1}+vx^{d_2})}
  =(V(u,v)-1)\cdot p^m,
\end{equation*}
where $V(u,v)$ is the number of common solutions to
  \begin{equation}\label{Vuv}
     \left \{\begin{array}{ll}
     \bar vz^{2g_1-1}+\bar uz^{g_1+g_2-1}+uz^{g_1-g_2}+v=0,\\
     z^{p^m+1}=1.\\
\end{array}\right.
  \end{equation}
\end{lemma}

\begin{lemma}\rm\label{luo wang}
\cite[Theorem 1]{luo111}
Let $p$ be a prime number, $n=2kl$ with $k$ and $l$ positive integers. Let  $lcm (n_1,n_2)\mid p^l+1$ and $t=\gcd(n_1,n_2)$. Let $0\leq r_1\leq n_1-1$ and $0\leq r_2\leq n_2-1$. Let $N_{p^n}(\chi)$ be the number of $\mathbb{F}_{p^n}$-rational points on the affine curve
\begin{equation*}
  \chi: \alpha x^{n_1}+\beta y^{n_2}+1=0.
\end{equation*}
 For $\alpha\in \left\{\psi^{r_1+n_1w}\mid 0\leq w\leq \frac{p^n-1}{n_1}-1 \right\}$ and
 $\beta\in\left\{\psi^{r_2+n_2w}\mid 0\leq w\leq \frac{p^n-1}{n_2}-1 \right\}$, we have the following results.
 \begin{itemize}
	\item[(i)] if $r_1=r_2=0$, then $N_{p^n}(\chi)=p^n+(-1)^{k-1}((n_1-1)(n_2-1)+1-t)p^{\frac{n}{2}}-t+1$.
	\item[(ii)] if $r_1=0$, $r_2\neq0$ and $t \nmid r_2$, then $N_{p^n}(\chi)=p^n+(-1)^k(n_1-2)p^{\frac{n}{2}}+1$.
	\item[(iii)] if $r_1\neq0$, $r_2=0$ and $t\nmid r_1$, then $N_{p^n}(\chi)=p^n+(-1)^k(n_2-2)p^{\frac{n}{2}}+1$.
	\item[(iv)] if $r_1\neq0$, $r_2\neq0$ and $t\nmid r_1-r_2$, then $N_{p^n}(\chi)=p^n+(-1)^{k-1}2p^{\frac{n}{2}}+1$.
	\item[(v)] if $r_1\neq0$, $r_2\neq0$ and $t\mid r_1-r_2$, then $N_{p^n}(\chi)=p^n+(-1)^{k}(t-2)p^{\frac{n}{2}}-t+1$.
\end{itemize}

\end{lemma}

It is well known that the two identities
\begin{equation}\label{two basic ide}
  \sum\limits_{i=0}^{\Delta_F}\omega_i=p^n \mbox{ and } \sum\limits_{i=0}^{\Delta_F}i\omega_i=p^n
\end{equation}
are useful for computing the differential spectrum of power function \(F\)(see \cite{firstdef}). However, just knowing these two identities may not be enough to calculate the differential spectrum. We also need the following results in Lemma \ref{fouride}.

\begin{lemma}\rm\label{fouride}
\cite[Theorem 10]{importentide}
With the above notations, let $M$ denote the number of solutions in $\left(\mathbb{F}_{p^n}\right)^4$ of
\begin{eqnarray*}
    \left \{
\begin{array}{lll}x_1-x_2+x_3-x_4=0,\\
x_1^{d}-x_2^{d}+x_3^d-x_4^{d}=0.\\
\end{array}\right.
\end{eqnarray*}
Then we have
\begin{equation}\label{very importent ide}
  \sum\limits_{i=0}^{\Delta_F}i^2\omega_i=\frac{M-p^{2n}}{p^n-1}.
\end{equation}
\end{lemma}

Note that the proof of Lemma \ref{fouride} is related to the Walsh transform which defined by the following definition.
\begin{definition}\rm
  Consider a function \( F : \mathbb{F}_{p^n} \rightarrow \mathbb{F}_{p^n} \). Its Walsh transform is given by
\begin{equation}\label{suv}
W_F(u,v)=\sum\limits_{x\in\mathbb{F}_{p^n}}\zeta_p^{{\rm{Tr}}_1^n(uF(x)-vx)},
\end{equation}
where $(u,v)\in \mathbb{F}_{p^n}\times\mathbb{F}_{p^n}$, $\zeta_p=e^{\frac{2\pi\sqrt{-1}}{p}}$ is a primitive complex $p$-th root of unity, and $\Tr^n_1(x)=\sum\limits_{i=0}^{n-1}x^{p^i}$ denotes the absolute trace function from $\mathbb{F}_{p^n}$ to $\mathbb{F}_{p}$.
The Walsh spectrum of the function \( F \) is subsequently defined as the multi-set
\[\left\{ W_F(u, v) :  \mbox{ } u \in \mathbb{F}_{p^n}^*, \mbox{ } v \in \mathbb{F}_{p^n} \right\}.\]
\end{definition}

Based on the above definition, we have the following results about Walsh spectrums.

\begin{lemma}\rm\label{powersum}
Let $N_r$ denote the number of solutions of

\begin{equation}\label{system of eq2}\left \{
\begin{array}{cll}x_1+x_2+\cdots+x_r=0,\\
x_1^{d}+x_2^{d}+\cdots+x_r^{d}=0,\\
\end{array}\right.
\end{equation}
in $\left(\mathbb{F}_{p^n}\right)^r$.  Then we have
\begin{equation}\label{eq in lem2}
\sum\limits_{u,v\in \mathbb{F}_{p^n} }W_{x^d}(u,v)^r=p^{2n}N_r.
\end{equation}
\end{lemma}
\begin{proof}
By definition, it is obvious that
\begin{equation*}
\begin{aligned}
    & \sum\limits_{u,v\in \mathbb{F}_{p^n} } W_{x^d}(u,v)^r \\
    &=\sum\limits_{u,v\in \mathbb{F}_{p^n}}\sum\limits_{x_1,x_2,\cdots,x_r\in \mathbb{F}_{p^n}}\zeta_p^{\Tr ^n_1\left(u(x_1^{d}+x_2^{d}+\cdots+x_r^{d})-v(x_1+x_2+\cdots+x_r)\right)}\\
    &=\sum\limits_{u,x_1,x_2,\cdots,x_r\in \mathbb{F}_{p^n}}\zeta_p^{\Tr ^n_1\left(u(x_1^{d}+x_2^{d}+\cdots+x_r^{d})\right)}\sum\limits_{v\in \mathbb{F}_{p^n}}\zeta_p^{\Tr ^n_1\left(v(x_1+x_2+\cdots+x_r)\right)}\\
    &=p^n\cdotp \sum\limits_{\substack {u,x_1,x_2,\cdots,x_r\in \mathbb{F}_{p^n}, \\x_1+x_2+\cdots+x_r=0}}\zeta_p^{\Tr ^n_1\left(u(x_1^{d}+x_2^{d}+\cdots+x_r^{d})\right)}\\
    &=p^{2n}\cdotp N_r.
    \end{aligned}
\end{equation*}
\end{proof}

\section{Differential, Walsh, Boomerang Spectra, and Cyclic Codes}\label{sec niho}

\subsection{Differential properties of Niho type power functions}
In this subsection, we first derive a closed-form expression for the differential uniformity of the power function \( F(x) = x^{s(p^m-1)+1} \). Furthermore, we completely characterize its differential spectrum when this power function is locally-APN. To this end, we need the results in Lemma \ref{crslemma}.

\begin{lemma}\label{crslemma}\rm
For any $i,j \in \{0, 1, \dots, p^m\}$, we have
\[
|C_{i,j}| =
\begin{cases}
p^m-2, & \text{if } i = j = 0, \\
1, & \text{if } i \neq j \text{ and } i j \neq 0, \\
0, & \text{if } i = j \neq 0 \text{ or } (i \neq j \text{ and } i j = 0).
\end{cases}
\]
\end{lemma}
\begin{proof}
(i) When $i = j = 0$, it is clear that $|C_{0,0}| = |\mathbb{F}_{p^{2m}}\setminus \{0,-1\}| = p^m-2$.

(ii) When $i \neq j$ and $ij \neq 0$, in order to compute $|C_{i,j}|$, we need to count the number of solutions to the system
\begin{equation}\label{001b}
    \left \{
\begin{array}{ll}
\psi^iy_1-\psi^jy_2=1,\\
\operatorname{ind}_{\psi}(y_1)=\operatorname{ind}_{\psi}(y_2)\equiv 0\,\,({\rm mod}\,\,p^m+1).
\end{array}\right.
\end{equation}
The condition $\operatorname{ind}_{\psi}(y_1)=\operatorname{ind}_{\psi}(y_2)\equiv 0\pmod{p^m+1}$ implies that we may set
$y_1 = \psi^{u(p^m+1)}$ and $y_2 = \psi^{v(p^m+1)}$, where $u, v \in \{0, 1, \dots, p^m-2\}$.
Substituting into $\psi^iy_1 - \psi^jy_2 = 1$, we obtain
\[
\psi^{i+u(p^m+1)} - \psi^{j+v(p^m+1)} = 1.
\]
Let $x = \psi^{u}$ and $y = \psi^{v}$. Then $x, y \in \mathbb{F}_{p^{2m}}^*$, and the equation becomes
\begin{equation}\label{bu1}
\psi^ix^{p^m+1} - \psi^jy^{p^m+1} = 1.
\end{equation}
Applying Lemma \ref{luo wang} (iv) with $n_1 = n_2 = p^m+1$, $k = 1$, $l = m$, $t = \gcd(p^m+1, p^m+1) = p^m+1$, $r_1 = i$, and $r_2 = j$, the number of $\mathbb{F}_{p^{2m}}$-rational points on the curve $\psi^ix^{p^m+1} - \psi^jy^{p^m+1} = 1$ is
\[
N = p^{2m} + (-1)^{1-1}(2p^m) + 1 = p^{2m} + 2p^m + 1 = (p^m+1)^2.
\]
Observe that if $(x_0, y_0)$ is a solution to \eqref{bu1}, then for any $k_1, k_2 \in \{0, 1, \dots, p^m\}$,
\[
(x_0 \psi^{k_1(p^m-1)}, y_0 \psi^{k_2(p^m-1)})
\]
is also a solution. Therefore, the number of solutions to the original system \eqref{001b} is
\[
\frac{N}{(p^m+1)^2} = \frac{(p^m+1)^2}{(p^m+1)^2} = 1.
\]
Hence, $|C_{i,j}| = 1$ for $i \neq j$ and $ij \neq 0$.

(iii) For the remaining cases, a similar analysis yields $|C_{i,j}| = 0$.
\end{proof}

\begin{theorem}\rm\label{niho theo}
Let $F(x)=x^{s(p^m-1)+1}$ be a Niho type power function over $\mathbb{F}_{p^n}$, where $n=2m$, $s_1=\gcd(s, p^m+1)$ and $s_2=\gcd(s-1, p^m+1)$. Then, the differential uniformity \( \triangle_F \) is given by
\[ p^m + (s_1 - 1)(s_1 - 2) + (s_2 - 1)(s_2 - 2).\]
Furthermore, if \(F\) is locally-APN, then its differential uniformity is $p^m$ and its differential spectrum is given by
\[ \mathbb{DS} = \left\{ \omega_0 = \frac{p^{2m} + p^m - 2}{2},\quad \omega_2 = \frac{p^{2m} - p^m}{2},\quad \omega_{p^m} = 1 \right\}.
\]
\end{theorem}

\begin{proof}
Let
\[\triangle(x)=(x+1)^{s(p^m-1)+1}-x^{s(p^m-1)+1}.\]
We get $\triangle(0)=1$ and $\triangle(-1)=1$. Hence, in the following we only need to consider the number of solutions of \(\Delta(x)=b\) in \(\mathbb{F}_{p^n}^{\sharp}\).
Let $\alpha=\psi^{p^m-1}$. Obviously,
$
 \bigsqcup\limits_{0\leq i,j\leq p^m} C_{i,j}=\mathbb{F}_{p^n}^{\sharp}
$
(here $\bigsqcup$ indicates disjoint union).  If $x\in\mathbb{F}_{p^n}^{\sharp}$, we consider $(p^m+1)^2$ cases:
\begin{table}[H]
\centering
\begin{tabular}{|l|c|c|c|c|c|}
\hline
$x$ in set & $\triangle(x)=b$ & $x+1$& $x$\\
[0.5ex]
\hline
 $C_{0,0}$ & $1=b$ & \makecell{To Be \\Determined} & \makecell{To Be \\Determined} \\
\hline
 $C_{0,1}$ & $x+1-\alpha^s x=b$ & $\frac{b-\alpha^s}{1-\alpha^s}$ & $\frac{b-1}{1-\alpha^s}$\\
\hline
 \vdots  & \vdots  & \vdots  & \vdots  \\
\hline
 $C_{0,p^m}$ & $x+1-\alpha^{p^ms} x=b$
 & $\frac{b-\alpha^{p^ms}}{1-\alpha^{p^ms}}$
& $\frac{b-1}{1-\alpha^{p^ms}}$\\
\hline
 $C_{1,0}$ & $\alpha^s(x+1)-x=b$
 & $\frac{b-1}{\alpha^s-1}$
& $\frac{b-\alpha^s}{\alpha^s-1}$\\
\hline
 $C_{1,1}$ & $\alpha^s=b$
 & \makecell{To Be \\Determined} & \makecell{To Be \\Determined} \\
\hline
 \vdots  & \vdots  & \vdots  & \vdots  \\
\hline
 $C_{1,p^m}$ & $\alpha^s(x+1)-\alpha^{p^ms}x=b$
 & $\frac{b-\alpha^{p^ms}}{\alpha^s-\alpha^{p^ms}}$
& $\frac{b-\alpha^s}{\alpha^s-\alpha^{p^ms}}$\\
\hline
 \vdots  & \vdots  & \vdots  & \vdots  \\
\hline
 $C_{p^m,0}$ & $\alpha^{p^ms}(x+1)-x=b$
 & $\frac{b-1}{\alpha^{p^ms}-1}$
& $\frac{b-\alpha^{p^ms}}{\alpha^{p^ms}-1}$\\
\hline
 \vdots  & \vdots  & \vdots  & \vdots  \\
\hline
 $C_{p^m,p^m}$ & $\alpha^{p^ms}=b$
 &  \makecell{To Be \\Determined} & \makecell{To Be \\Determined} \\
\hline
\end{tabular}
\end{table}
{\textbf{Case 1:}}
$x\in C_{0,0}$. By $|C_{0,0}|=p^m-2$ and the table above, it is easy to see that $\triangle(x) = 1$ has exactly $p^m - 2$ solutions in $C_{0,0}$, and $\triangle(x) = b$ has no solution in $C_{0,0}$ for $b \neq 1$.

{\textbf{Case 2:}}
$x\in \left(\bigsqcup\limits_{0< i\leq p^m} C_{i,i}\right) \bigsqcup \left(\bigsqcup\limits_{0 \le i,j \le p^m, i \neq j, ij=0} C_{i,j}\right)$.
By Lemma \ref{crslemma}, both unions have cardinality $0$.
Therefore, for any $b\in \mathbb{F}_{p^n}$, $\Delta(x) = b$ has no solution in this set.

{\textbf{Case 3:}}
$x\in \bigsqcup\limits_{1\leq i,j\leq p^m, i\neq j} C_{i,j}$.
In this case, if $\triangle(x)=b$ has one solution $x_1\in C_{j_1,j_2}$ and one solution $x_2\in C_{j_1,j_3}$ with $j_2\neq j_3$ and $j_1\neq j_2,j_3$, then
\begin{equation*}
  x_1=\frac{b-\alpha^{j_1s}}{\alpha^{j_1s}-\alpha^{j_2s}}, x_1+1=\frac{b-\alpha^{j_2s}}{\alpha^{j_1s}-\alpha^{j_2s}},
\end{equation*}
\begin{equation*}
  x_2=\frac{b-\alpha^{j_1s}}{\alpha^{j_1s}-\alpha^{j_3s}}, x_2+1=\frac{b-\alpha^{j_3s}}{\alpha^{j_1s}-\alpha^{j_3s}}.
\end{equation*}
Notice that $\alpha^i\in \mu_{p^m+1}$ for $0\leq i \leq p^m$ which implies $\alpha^{ip^m}=\alpha^{-i}$. Then we can get
\begin{equation}\label{x1pl11}
  x_1^{p^m-1}=\left(\frac{b-\alpha^{j_1s}}{\alpha^{j_1s}-\alpha^{j_2s}}\right)^{p^m-1}
  =\frac{(1-\alpha^{j_1s}b^{p^m})\alpha^{j_2s}}{b-\alpha^{j_1s}},
\end{equation}
\begin{equation}\label{x1+1pl11}
  (x_1+1)^{p^m-1}=\left(\frac{b-\alpha^{j_2s}}{\alpha^{j_1s}-\alpha^{j_2s}}\right)^{p^m-1}
  =\frac{(1-\alpha^{j_2s}b^{p^m})\alpha^{j_1s}}{b-\alpha^{j_2s}},
\end{equation}
\begin{equation}\label{x2pl11}
  x_2^{p^m-1}=\left(\frac{b-\alpha^{j_1s}}{\alpha^{j_1s}-\alpha^{j_3s}}\right)^{p^m-1}
  =\frac{(1-\alpha^{j_1s}b^{p^m})\alpha^{j_3s}}{b-\alpha^{j_1s}},
\end{equation}
and
\begin{equation}\label{x2+1pl11}
  (x_2+1)^{p^m-1}=\left(\frac{b-\alpha^{j_3s}}{\alpha^{j_1s}-\alpha^{j_3s}}\right)^{p^m-1}
  =\frac{(1-\alpha^{j_3s}b^{p^m})\alpha^{j_1s}}{b-\alpha^{j_3s}}.
\end{equation}
Meanwhile, from \(x_1\in C_{j_1,j_2}\) and \(x_2\in C_{j_1,j_3}\) we obtain
\[
x_1=\psi^{k_1(p^m+1)+j_2},\quad x_1+1=\psi^{k_2(p^m+1)+j_1},
\]
\[
x_2=\psi^{k_3(p^m+1)+j_3},\quad x_2+1=\psi^{k_4(p^m+1)+j_1},
\]
with \(0\le k_i\le p^m-2\) for \(1\le i \le 4\).
Therefore,
\begin{equation}\label{x1alpha}
x_1^{p^m-1} = \alpha^{j_2}, \quad (x_1+1)^{p^m-1} = \alpha^{j_1},
\end{equation}
\begin{equation}\label{x2alpha}
  x_2^{p^m-1} = \alpha^{j_3}, \quad (x_2+1)^{p^m-1} = \alpha^{j_1}.
\end{equation}
Combining (\ref{x1pl11})-(\ref{x2alpha}), we deduce
\begin{equation}\label{x1pl}
 \frac{(1-\alpha^{j_1s}b^{p^m})\alpha^{j_2s}}{b-\alpha^{j_1s}}=\alpha^{j_2},
\end{equation}
\begin{equation}\label{x1+1pl}
\frac{(1-\alpha^{j_2s}b^{p^m})\alpha^{j_1s}}{b-\alpha^{j_2s}}=\alpha^{j_1},
\end{equation}
\begin{equation}\label{x2pl}
\frac{(1-\alpha^{j_1s}b^{p^l})\alpha^{j_3s}}{b-\alpha^{j_1s}}=\alpha^{j_3},
\end{equation}
and
\begin{equation}\label{x2+1pl}
\frac{(1-\alpha^{j_3s}b^{p^m})\alpha^{j_1s}}{b-\alpha^{j_3s}}=\alpha^{j_1}.
\end{equation}

\noindent\textbf{Subcase 1: }
$b=1$. By (\ref{x1pl}) and (\ref{x1+1pl}), we have
\[ \alpha^{j_1(1-s)} = \alpha^{j_2(1-s)} = 1. \]
In the cyclic group of order \( p^m+1 \), the equation \( x^{1-s} = 1 \) has exactly \( s_2 = \gcd(s-1, p^m+1) \) solutions. Hence, there exist \( s_2-1 \) distinct elements \( j_1,\dots,j_{s_2-1} \) in \( (0, p^m] \) such that \( \alpha^{j_k(1-s)} = 1 \) for \( 1 \le k \le s_2-1 \). Consequently, for any \( k \neq l \) with \( 1 \le k,l \le s_2-1 \), \( \Delta(x) = 1 \) has exactly one solution in \( C_{j_k,j_l} \).
Similarly, the equation \( x^s = 1 \) has \( s_1 = \gcd(s, p^m+1) \) solutions in the cyclic group of order \( p^m+1 \), giving \( s_1-1 \) distinct values \( t_1,\dots,t_{s_1-1} \) in \( (0, p^m] \) with \( \alpha^{t_i s} = 1 \). Therefore, \( \Delta(x) = 1 \) also has exactly one solution in each \( C_{t_i,t_j} \) for \( 1 \le i,j \le s_1-1 \) with \( i \neq j \). Therefore, the total number of solutions to $\triangle(x)=1$ in
\[\bigsqcup\limits_{1\leq i,j\leq p^m, i\neq j} C_{i,j}\]
 is $(s_1 - 1)(s_1 - 2) + (s_2 - 1)(s_2 - 2)$.

\noindent\textbf{Subcase 2: }
$b\in\mu_{p^m+1}\setminus\{1\}$.   By (\ref{x1pl}) and (\ref{x1+1pl}),
\[ \alpha^{j_1(s-1)} = \alpha^{j_2(s-1)} = b. \]
Since \( s_2 = \gcd(s - 1, p^m + 1) \), for \( b \in \langle \alpha^{s_2} \rangle \), the equation \( x^{1-s} = b \) yields \( s_2 \) distinct solutions \( \alpha^{j_1}, \dots, \alpha^{j_{s_2}} \). Hence, for any distinct \( k, l \) with \( 1 \le k, l \le s_2 \), \( \Delta(x) = b \) has exactly one solution in \( C_{j_k, j_l} \). Similarly,  if \( b \in \langle \alpha^{s_1} \rangle \), there exist \( s_1 \) distinct elements \( t_1,\dots,t_{s_1} \) in \( (0, p^m] \) satisfying \( \alpha^{t_i s} = b \). Consequently, for any distinct \( i, j \) with \( 1 \le i,j \le s_1 \), \( \Delta(x) = b \) has exactly one solution in \( C_{t_i,t_j} \).  Hence, for \(b \in \langle \alpha^{s_1s_2} \rangle\) (where \(s_1s_2 < p^m+1\)), \(\triangle(x)=b\) has \(s_1(s_1 - 1) + s_2(s_2 - 1)\) solutions in
\[\bigsqcup\limits_{1\leq i,j\leq p^m, i\neq j} C_{i,j}\].

\noindent\textbf{Subcase 3: }
$b\not\in\mu_{p^m+1}$.
By (\ref{x1pl}) and (\ref{x2pl}), we have
\[\alpha^{j_2(1-s)}=\alpha^{j_3(1-s)}.\]
By (\ref{x1+1pl}) and (\ref{x2+1pl}), we can get
\[(\alpha^{j_2s}-\alpha^{j_3s})(1-b^{p^l+1})=0.\]
But this contradict  ${j_2\neq j_3}$. Therefore, for any fixed \( j \) with \( 1 \leq j \leq p^m \), the equation \(\Delta(x) = b\) has at most one solution in \[\bigsqcup\limits_{1 \leq i \leq p^m} C_{j,i}.\]
Therefore, for \( b \notin \mu_{p^m+1} \), the equation \(\Delta(x) = b\) has at most \( p^m \) solutions in
\[
\bigsqcup\limits_{\substack{1 \leq i,j \leq p^m \\ i \neq j}} C_{i,j}.
\]

Based on the above discussion, \(\Delta(x) = 1\) has
 \[p^m + (s_1 - 1)(s_1 - 2) + (s_2 - 1)(s_2 - 2)\]
 solutions.
When \( b \in \mu_{p^m+1} \setminus \{1\} \), the equation \( \Delta(x) = b \) has at most
\[
s_1(s_1 - 1) + s_2(s_2 - 1)
\]
solutions. As shown previously, we have \( s_1 s_2 < p^m + 1 \) in this case. By divisibility, \( s_1 s_2 \) is a proper divisor of \( p^m + 1 \), hence
\[
s_1 s_2 \le \frac{p^m + 1}{2},
\]
 which yields
\[
s_1 + s_2 \le \frac{p^m + 1}{2} + 1 < \frac{p^m + 4}{2}.
\]
Expanding and rearranging this inequality gives
\[
s_1(s_1 - 1) + s_2(s_2 - 1) < p^m + (s_1 - 1)(s_1 - 2) + (s_2 - 1)(s_2 - 2).
\]
For \(b \notin \mu_{p^m+1}\), \(\Delta(x) = b\) has at most \(p^m\) solutions. Hence,
\[
\triangle_F = p^m + (s_1 - 1)(s_1 - 2) + (s_2 - 1)(s_2 - 2).
\]
It is also noted that when \(F\) is locally-APN, i.e.,
\[
\max\bigl\{ \triangle_F(1,b) \mid b \in \mathbb{F}_{p^n} \setminus \mathbb{F}_p \bigr\} = 2,
\]
we must have \(s_1(s_1-1)+s_2(s_2-1) \le 2\), and hence
\[
(s_1-1)(s_1-2)+(s_2-1)(s_2-2)=0.
\]
Moreover, if \(x\) is a solution of \(\Delta(x)=b\), then \(-1-x\) is also a solution. This observation gives \(\omega_1 = 0\). Consequently, when \(F\) is locally-APN, \(\triangle_F = p^m\) and its differential spectrum can be derived from Equation (\ref{two basic ide}). This completes the proof.
\end{proof}

The following result is useful for determining the Walsh spectrum of Niho type power functions.
\begin{corollary}\rm\label{third sum 1}
  Let $F(x)=x^d=x^{s(p^m-1)+1}$ be a Niho type power function defined over $\mathbb{F}_{p^n}$, where $n=2m$, $s_1=\gcd(s, p^m+1)$ and $s_2=\gcd(s-1, p^m+1)$. Then
  \makecell{ $\sum\limits_{u,v\in \mathbb{F}_{p^n} }W_{F}(u,v)^3 =$~~~~~~~~~~~~~~~~~~~~~~~~~~~~~~~~~~~~~~~~ \\~~~~~~~~~  $p^{2n}\big(3\cdotp p^{n}-2+(p^n-1)(p^m + (s_1 - 1)(s_1 - 2) $\\ ~~~~~~~~~~~~~~~~~~~~~~~~~~~~~~~~~~~~~~ $+(s_2 - 1)(s_2 - 2)-2)\big).$}
\end{corollary}
\begin{proof}
We only need to count the number of solutions of the system
\begin{equation}\label{third111}\left \{
\begin{array}{cll}x_1+x_2+x_3=0,\\
x_1^{d}+x_2^{d}+x_3^{d}=0,\\
\end{array}\right.
\end{equation}
in $\left(\mathbb{F}_{p^n}\right)^3$. If $x_1x_2x_3=0$, we can assume that $x_3=0$, which provides $p^{n}$ solutions of (\ref{third111}): $x_1=x_2$, $x_3=0$. In total, there are $3\cdotp p^{n}-2$ solutions of (\ref{third111}) satisfying $x_1x_2x_3=0$, because $(0,0,0)$ is counted triple. If $x_1x_2x_3\neq0$, then the number of solutions of (\ref{third111}) is $(p^{n}-1)\cdotp N_1'$, where $N_1'$ is the number of solutions to
\begin{equation}\label{n3' equation}
    \left \{
\begin{array}{cll}y_1+y_2=1,\\
y_1^{d}+y_2^{d}=1,\\
\end{array}\right.
\end{equation}
satisfying $y_1y_2\neq0$. It follows from Theorem \ref{niho theo} that
\[N_1' =p^m + (s_1 - 1)(s_1 - 2) + (s_2 - 1)(s_2 - 2)-2.\]
This completes the proof.
\end{proof}

\begin{remark}
  \begin{enumerate}
    \item[(1).] For the power function \(F(x) = x^d\) defined over \(\mathbb{F}_{p^n}\), observe that when \(\gcd(d, p^n - 1) = 1\), for any fixed \(u \in \mathbb{F}_{p^n}^*\), the value distribution of \(W_F\bigl(1, \, v / u^{d^{-1}}\bigr)\) as \(v\) runs through \(\mathbb{F}_{p^n}^*\) coincides with that of \(W_F(1, v)\) as \(v\) runs through \(\mathbb{F}_{p^n}^*\), where \(d^{-1}\) denotes the multiplicative inverse of \(d\) modulo \(p^n - 1\). Hence, if \(\gcd(d, p^n - 1) = 1\),  the value distribution of \(W_F(u, v)\) while \((u, v)\) ranges over \(\mathbb{F}_{p^n}^* \times \mathbb{F}_{p^n}\) is transformed into studying the value distribution of \(W_F(1, v)\) as \(v\) ranges over \(\mathbb{F}_{p^n}\), which is essentially cross-correlation distribution.

    \item[(2).]  The value distribution of cross-correlations for the Niho type power function \(F(x)=x^{s(p^m-1)+1}\) over \(\mathbb{F}_{p^n}\) has been extensively studied. For instance,
    \begin{itemize}
    \item \(s=2\) with \(p^m\not\equiv 2\pmod{3}\) \cite{H2};
    \item \(s=3\) with \(p=2\) and \(\gcd(3\cdot 2^{m}-2,2^n-1)=1\) \cite{xia1};
    \item \(s=3\) with \(p=3\) and \(m\not\equiv 2\pmod{4}\) \cite{xia2};
    \item \(s=3\) with \(p \geq 5\) and \(\gcd(3\cdot p^{m}+2,p^n-1)=1\) \cite{xiong}.
    \end{itemize}
It is worth noting that determining the above cross-correlation value distributions requires computing \(\sum_{v\in\mathbb{F}_{p^n}}W_{F}(1,v)^3\), which can be promptly obtained by Corollary \ref{third sum 1}.

    \item[(3).] In 1972, Niho conjectured \cite{Niho1972} that when \(m\) is even, the cross-correlation of \(x^{4(2^m-1)+1}\) over \(\mathbb{F}_{2^{2m}}\) is at most five-valued. This conjecture was proved by Helleseth, Katz, and Li in 2021 \cite{s4}. It is therefore natural to proceed with determining the value distribution of this cross-correlation. Through Corollary \ref{third sum 1}, a fourth equation concerning the value distribution can be derived, so only one additional equation is needed to fully determine the distribution.
\end{enumerate}
\end{remark}

\subsection{Walsh spectrum and cyclic codes of Niho type locally-APN power functions}
In this subsection, we will characterize the equivalence between locally-APN and four-valued Walsh transforms  for power functions of Niho type. Equivalently, the corresponding cyclic codes have exactly four nonzero weights. Before stating the main result, we need to introduce some auxiliary tools.

Let $h_1(x)$ and $h_d(x)$ be the minimal polynomials of $\psi^{-1}$ and $\psi^{-d}$ over $\mathbb{F}_{p}$, respectively. Let $\mathcal{C}_{1,d}$ be the cyclic code with parity-check polynomial $h_1(x)h_d(x)$. By Delsarte's Theorem \cite{Delsarte_theorem}, the cyclic code $\mathcal{C}_{1,d}$ can be expressed as
\begin{equation*}
    \mathcal{C}_{1,d}=\left\{c_{u,v}=\left(\Tr ^{n}_1(u\psi^{id}+v\psi^i)\right)_{i=0}^{p^n-2}\mid u,v\in \mathbb{F}_{p^{n}} \right\}.
\end{equation*}

From Theorem \ref{niho theo} and Equation (\ref{very importent ide}), the result in Corollary \ref{four sum 1} is  easily derived and we omit its proof. This will be useful in the proof of Theorem \ref{nihowalsh}.
\begin{corollary}\rm\label{four sum 1}
  Let $F(x)=x^{d}$ be a Niho type locally-APN function defined over $\mathbb{F}_{p^n}$, where $d=s(p^m-1)+1$ and $n=2m$. Then the number of solutions $(x_1,x_2,x_3,x_4)\in\mathbb{F}_{p^n}^4$ of the system of equations
\begin{equation*}
\left \{
\begin{array}{cll}x_1+x_2+x_3+x_4=0,\\
x_1^{d}+x_2^{d}+x_3^d+x_4^{d}=0,\\
\end{array}\right.
\end{equation*}
is $4p^{2n}-2p^{3m}-3p^n+2p^m$.
\end{corollary}

Next we give our main result.

\begin{theorem}\rm\label{nihowalsh}
Let \(F(x)=x^d=x^{s(p^m-1)+1}\) be a Niho type power function over \(\mathbb{F}_{p^n}\) with \(n=2m\). Then the following three statements are equivalent.

(1) \(F\) is locally-APN.

(2) The Walsh spectrum of \(F\) takes four values in \(\{-p^m, 0, p^m, 2p^m\}\). Moreover, when \((u,v)\) runs through \(\mathbb{F}_{p^{n}}^* \times \mathbb{F}_{p^{n}}\), the value distribution of \( W_F(u,v) \) is given by
\begin{equation}\label{WF}
W_F(u, v) =
\left\{\begin{array}{llll}
-p^m, & \text{occurs } \frac{p^{4m}-p^{3m}-p^{2m}+p^m}{3} \text{ times}, \\[6pt]
0, & \text{occurs } \frac{p^{4m}-p^{3m}-p^{2m}+p^m}{2} \text{ times}, \\[6pt]
p^m, & \text{occurs } p^{3m}-p^m\text{ times}, \\[6pt]
2p^m, & \text{occurs } \frac{p^{4m}-p^{3m}-p^{2m}+p^m}{6} \text{ times}.
  \end{array}\right.
\end{equation}

(3)  The cyclic codes $\mathcal{C}_{1,d}$ have four nonzero weights: $p^{m-1}(p-1)(p^m + k)$ for $k = 0, 1, -1, -2$.
\end{theorem}

\begin{proof}
Denote \( s_1 = \gcd(s, p^m+ 1) \), \( s_2 = \gcd(s - 1, p^m + 1) \) and \( \mu_{p^m+1}^i  = \{x^{s_i}: x \in \mu_{p^m+1}\} \) for \( i = 1, 2 \). Note that \( \mu_{p^m+1}^i \) are subgroups of \( \mu_{p^m+1} \). Put $d_1=s(p^m-1)+1$ and $d_2=1$ in Lemma \ref{niho the}, we have
\begin{equation*}
  W_F(u,-v)=(V(u,v)-1)\cdot p^m,
\end{equation*}
where $V(u,v)$ is the number of common solutions to
  \begin{equation}\label{Vuv1}
     \left \{\begin{array}{ll}
     \bar vz^{2s-1}+\bar uz^{s-1}+uz^{s}+v=0,\\
     z^{p^m+1}=1.\\
\end{array}\right.
  \end{equation}
Denote by
\begin{equation*}
  E_i=\#\{(u,v) \in \mathbb{F}_{p^n}^2\setminus(0,0):  V(u,v)=i \}.
\end{equation*}
Furthermore, let \( V_1(u,v) \) be the number of solutions to (\ref{Vuv1}) corresponding to \( u,v \in \mu_{p^m+1} \). From (\ref{Vuv1}) we obtain in the case \( u,v \in \mu_{p^m+1} \) that

\[
\bar vz^{2s-1}+\bar uz^{s-1}+uz^{s}+v= \overline{v}(z^{s-1} + uv)(z^s + \overline{u}v) = 0.
\]
Therefore, we get
\begin{itemize}
    \item[(i)] \( V_1(u,-1) = s_1 \) for \( u \in \mu_{p^m+1}^1 \setminus \mu_{p^m+1}^2 \),
    \item[(ii)] \( V_1(u,-1) = s_2 \) for \( u \in \mu_{p^m+1}^2 \setminus \mu_{p^m+1}^1 \),
    \item[(iii)] \( V_1(u,-1) = s_1 + s_2 \) for \( u \in \mu_{p^m+1}^1 \cap \mu_{p^m+1}^2 \), \( u \neq 1 \),
    \item[(iv)] \( V_1(1,-1) = s_1 + s_2 - 1 \).
\end{itemize}

\textbf{Proof of (1) $\Rightarrow$ (2).}
According to Theorem 1, \(s_1, s_2 \in \{1, 2\}\) when \(F\) is locally-APN. We therefore break down the discussion into the following two cases.

{\textbf{Case 1:}} When \( s_1 = s_2 = 1 \), it follows from (iii) that \( V(u,-1) = 2 \) for any \( u \in \mu_{p^m+1} \setminus \{1\} \), and from (iv) that \( V(1,-1) = 1 \). Therefore, we have \( E_1 > 0 \) and \( E_2 > 0 \), which implies that the Walsh transform is at least two-valued. We assume that it is $(k + 3)$-valued (ignoring \( W_F(0,0)=p^{2m} \)). Then, by Lemma \ref{powersum} and Corollaries \ref{third sum 1} and \ref{four sum 1}, we obtain the five equations as follows:

\begin{equation}\label{1 eq1}
  E_{0}+E_{1}+E_2+\sum_{j=1}^k E_{t_j+1}=p^{2n}-1,
\end{equation}
\begin{equation}\label{1 eq2}
  -E_{0}+E_2+\sum_{j=1}^k t_j E_{t_j+1} =p^{3m}-p^m,
\end{equation}
\begin{equation}\label{1 eq3}
 E_{0}+E_2+\sum_{j=1}^k t_j^2 E_{t_j+1}=p^{2n}-p^{n},
\end{equation}
\begin{equation}\label{1 eq4}
 -E_{0}+E_2+\sum_{j=1}^k t_j^3 E_{t_j+1}=p^{2n}-p^{n},
\end{equation}
\begin{equation}\label{1 eq5}
 E_{0}+E_2+\sum_{j=1}^k t_j^4 E_{t_j+1}=3p^{2n}-2p^{3m}-3p^n+2p^m,
\end{equation}
where \( 2 \leq t_1 < t_2 < \cdots < t_k \).
By (\ref{1 eq3}) and (\ref{1 eq4}), we obtain
\begin{equation}\label{E01}
  E_0 = \frac12 \sum_{j=1}^k t_j^2(t_j - 1) E_{t_j+1} .
\end{equation}
Subtracting  (\ref{1 eq2}) from (\ref{1 eq3}) yields
\begin{equation}\label{E02}
2E_0 + \sum_{j=1}^k (t_j^2 - t_j) E_{t_j+1} = p^{2n} - p^{3m} - p^n + p^m .
\end{equation}
Substituting (\ref{E01}) into (\ref{E02}) then gives
\begin{equation}\label{E03}
\sum_{j=1}^k (t_j^3 - t_j) E_{t_j+1} = p^{2n} - p^{3m} - p^n + p^m .
\end{equation}
Next, subtracting (\ref{1 eq3}) from (\ref{1 eq5}) yields
\begin{equation}\label{E04}
\frac12 \sum_{j=1}^k (t_j^4 - t_j^2) E_{t_j+1} = p^{2n} - p^{3m} - p^n + p^m .
\end{equation}
Finally, (\ref{E03}) and (\ref{E04}) leads to
\[
\sum_{j=1}^k t_j (t_j - 2)(t_j - 1)(t_j + 1) E_{t_j+1} = 0 .
\]
It follows that \( E_{t_i+1} = 0 \) for all \( 2\leq i \leq k \). If \( t_1 = 2 \) and \( E_3 = 0 \), then from (\ref{E01}) we obtain \( E_0 = 0 \). In this case, using (\ref{1 eq2}) and (\ref{1 eq3}) again, \( E_2 \) would simultaneously equal \( p^{3m} - p^m \) and \( p^{2n} - p^n \), a contradiction. Hence, \( E_0 \neq 0 \) and \( E_3 \neq 0 \). Therefore, only \( E_0, E_1, E_2, E_3 \) are nonzero, which means that the Walsh spectrum of \(F\) takes exactly the four values \(-p^m, 0, p^m, 2p^m\). Then, by (\ref{1 eq1})-(\ref{1 eq4}), we obtain the value distribution of \( W_F(u,v) \).

{\textbf{Case 2:}}  In the case that \( s_1 = 1, s_2 = 2 \) or \( s_1 = 2, s_2 = 1 \), we deduce from (iv) that \( E_2 > 0 \), and from (i) and (ii) that \( E_1 > 0 \). Furthermore, (iii) implies \( E_3 > 0 \). This  means that the Walsh transform is at least three-valued. Assuming it is \((k + 4)\)-valued (ignoring \( W_F(0,0)=p^{2m} \)), we thereby obtain the following system of five equations

\begin{equation}\label{2 eq1}
  E_{0}+E_{1}+E_2+E_3+\sum_{j=1}^k E_{t_j+1}=p^{2n}-1,
\end{equation}
\begin{equation}\label{2 eq2}
  -E_{0}+E_2+2E_3+\sum_{j=1}^k t_j E_{t_j+1}=p^{3m}-p^m,
\end{equation}
\begin{equation}\label{2 eq3}
 E_{0}+E_2+4E_3+\sum_{j=1}^k t_j^2 E_{t_j+1}=p^{2n}-p^{n},
\end{equation}
\begin{equation}\label{2 eq4}
 -E_{0}+E_2+8E_3+\sum_{j=1}^k t_j^3 E_{t_j+1}=p^{2n}-p^{n},
\end{equation}
\begin{equation}\label{2 eq5}
 E_{0}+E_2+16E_3+\sum_{j=1}^k t_j^4 E_{t_j+1}=3p^{2n}-2p^{3m}-3p^n+2p^m,
\end{equation}
where \( 3 \leq t_1 < t_2 < \cdots < t_k \). Using an approach similar to that in Case 1, we obtain
\[
\sum_{j=1}^k t_j (t_j - 2)(t_j - 1)(t_j + 1) E_{t_j+1} = 0 .
\]
It follows that \( E_{t_i+1} = 0 \) for all \( 1\leq i \leq k \). If \( E_0 = 0 \), a contradiction follows from (\ref{2 eq3}) and (\ref{2 eq4}). Therefore, only \( E_0, E_1, E_2, E_3 \) remain nonzero.

\textbf{Proof of (2) $\Rightarrow$ (1).}
If \( W_F(u, v) \) takes values in \(\{-p^m, 0, p^m, 2p^m\}\),   then from (i)-(iv) we can still deduce that \( s_1, s_2 \in \{1, 2\} \), and hence \( \triangle_F = p^m \). For each fixed \( b \in \mathbb{F}_{p^n} \), suppose the number of solutions to the equation
\[
(x+1)^{s(p^m-1)+1} - x^{s(p^m-1)+1} = b
\]
takes values in \( \{0, t_1, \ldots, t_k, p^m\} \), where \( 0 < t_1 < \cdots < t_k < p^m \). Then, using (\ref{two basic ide}), (\ref{very importent ide}) and Corollary \ref{four sum 1}, we obtain the following system of equations
\begin{numcases}{}
\omega_0+\omega_{t_1}+\cdots+\omega_{t_k}+1=p^n    \label{1w} \\
{t_1}\omega_{t_1}+\cdots+{t_k}\omega_{t_k}+p^m=p^n  \label{2w}  \\
 {t_1^2}\omega_{t_1}+\cdots+{t_k^2}\omega_{t_k}+p^n=3p^{n}-2p^{m} \label{3w}.
\end{numcases}
By subtracting (\ref{2w}) from (\ref{3w}), we obtain
\begin{equation}\label{wb}
\sum_{j=1}^k t_j(t_j - 1) \omega_{t_j} = p^n-p^m .
\end{equation}
Then, subtracting (\ref{wb}) from (\ref{2w}) yields
\[
\sum_{j=1}^kt_j(t_j - 2) \omega_{t_j} = 0 .
\]
Therefore, we obtain $\omega_{2} \neq 0$, and for all $2 < i < p^m$, we have $\omega_{i} = 0$, which implies that $F$ is locally-APN. This completes the proof.

\textbf{Proof of  (2) $\Leftrightarrow$ (3).} For a codeword $c_{u,v}$ of $\mathcal{C}_{1,d}$, we use exponential sums to
calculate its Hamming weight
\begin{equation*}
    \begin{aligned}
        \omega_H(c_{u,v}) &=p^{n}-1-\# \left\{x\in\mathbb{F}_{p^{n}}^*\mid \Tr_1^{n}(ux^{d}+vx)=0 \right\}\\
        &=p^{n}-\frac{1}{p}\sum\limits_{y\in \mathbb{F}_{p}}\sum\limits_{x\in \mathbb{F}_{p^{n}}}\omega_p^{y{\Tr ^{n}_1(ux^{d}+vx)}}\\
        &=p^{n-1}(p-1)-\frac{1}{p}\sum\limits_{y\in \mathbb{F}_{p}^*}\sum\limits_{x\in \mathbb{F}_{p^{n}}}\omega_p^{\Tr ^{n}_1(u{(xy)}^{d}+vxy)}\\
        &=p^{n-1}(p-1)-\frac{p-1}{p}\sum\limits_{x\in \mathbb{F}_{p^{n}}}\omega_p^{\Tr ^{n}_1(u{x}^{d}+vx)}.\\
    \end{aligned}
\end{equation*}
Therefore, the weight distribution of the cyclic code $\mathcal{C}_{1,d}$ is completely determined by the value distribution of the Walsh transform
\[
W_{x^d}(u,v)=\sum_{x\in\mathbb{F}_{p^n}}\zeta_p^{{\rm{Tr}}_1^n(ux^d-vx)}.
\]
If $u=0$, we have
\begin{equation*}
    \omega_H(c_{0,v})=
	\left\{ \begin{array}{llll}
		0&\mbox{if}&v=0,\\
		p^{n-1}(p-1)&\mbox{if}&v\neq0.
	\end{array}\right.
\end{equation*}
If $u\neq0$, the value distribution of $\omega_H(c_{u,v})$ follows from the Walsh spectrum of $F$.
\end{proof}

\begin{remark}\rm\label{F1F2}
Consider the following two families of power functions over $\mathbb{F}_{p^n}$ with $n = 2m$.
\begin{itemize}
    \item \(F_1(x) = x^{d_1}\) over $\mathbb{F}_{2^n}$, where \(d_1 = s(2^m - 1) + 1\) with
          \[
          s = \begin{cases}
              2^r \cdot (2^r - 1)^{-1}, & \text{if } \gcd(2^r - 1, 2^m + 1) = 1,\\[4pt]
              2^r \cdot (2^r + 1)^{-1}, & \text{if } \gcd(2^r + 1, 2^m + 1) = 1,
          \end{cases}
          \]
          under the conditions \(r < m\) and \(\gcd(r, m) = 1\);
    \item \(F_2(x) = x^{d_2}\) over $\mathbb{F}_{p^n}$ for any prime $p$, where \(d_2 = 2(p^m - 1) + 1\).
\end{itemize}
It was shown in \cite{via, H2, xias=2} that both $F_1(x)$ and $F_2(x)$ have four-valued Walsh spectrum contained in $\{-p^m, 0, p^m, 2p^m\}$.
By Theorem \ref{nihowalsh}, it follows immediately that $F_1(x)$ and $F_2(x)$ are locally-APN. Conversely, consider the function \(F_3(x) = x^{s(2^m - 1) + 1}\) over $\mathbb{F}_{2^{2m}}$, where $\gcd(2^k + 1, 2^m + 1) = 1$, $s = (2^k + 1)^{-1}$ modulo $2^m + 1$, and $\gcd(k, m) = 1$. Both $F_2(x)$ and $F_3(x)$ are known to be locally-APN in \cite{TIT.he,yan.niho}. By Theorem \ref{nihowalsh}, it follows immediately that both functions have four-valued Walsh spectrum.

It is suggested that the function $F_1(x)$ might include all Niho type power functions whose Walsh spectrum consists of $\{-2^m, 0, 2^m, 2^{m+1}\}$  as a conjecture in \cite{via}. However, by Theorem \ref{nihowalsh}, this is equivalent to the statement that $F_1(x)$ covers all Niho type locally-APN power functions.
Thus, classifying Niho type locally-APN power functions would resolve the conjecture of \cite{via}, and conversely, proving the conjecture would yield a complete classification of Niho type locally-APN power functions.
This equivalence motivates the following conjecture.
\end{remark}

\begin{Conjecture}\rm
The function $F_1(x) = x^{s(2^m - 1) + 1}$ over $\mathbb{F}_{2^{2m}}$, with $s$ defined as above, covers all Niho type locally-APN power functions over $\mathbb{F}_{2^{2m}}$.
\end{Conjecture}

\subsection{Feistel boomerang properties of Niho type locally-APN power functions}\label{sec fei}

Throughout this subsection, we always assume that \(F(x) = x^{d_3}\) is a Niho type locally-APN defined over \(\mathbb{F}_{p^n}\), where $d_3=s(p^m - 1) + 1$ and \(n = 2m\).
We focus on investigating two properties that their Feistel Boomerang Connectivity Table (in the case $p=2$) and second-order zero differential spectra (in the case $p>2$).

\begin{theorem}\rm\label{th11}
For any $(a,b)\in\mathbb{F}_{2^n}^2$, the value distribution of the multi-set $\left\{{\rm FBCT}_F(a,b): a,b \in \mathbb{F}_{2^n} \right\}$ is shown in Table \ref{fbct22}.
\begin{table}[H]
\centering
 \caption{ Feistel Boomerang Connectivity Table}
\label{fbct22}
\begin{tabular}{|c|c|c|}
 \hline
  ${\rm FBCT}_F(a,b)$ &  $(a,b)$  &{\rm Multiplicity} \\\hline
    $2^n$  & $ab(a+b)=0$  &  $3\cdot2^n-2$  \\ \hline
    $2^m$ & $ab(a+b)\neq0$, $\frac{a}{b}\in \mathbb{F}_{2^m}^{\sharp}$  & $(2^m-2)(2^n-1)$ \\ \hline
    0  &  {\rm otherwise}  & $(2^n-2^m)(2^n-1)$ \\\hline
\end{tabular}
\end{table}
\end{theorem}

\begin{proof}
We mainly study the number of solutions $x\in\mathbb{F}_{2^n}$ of
\begin{equation}\label{FBCT2f2}
  (x+a+b)^{d_3}+(x+a)^{d_3}
  +(x+b)^{d_3}+x^{d_3}=0
\end{equation}
for any $(a,b)\in \mathbb{F}_{2^n}^2$.

{\textbf{Case 1:}} $ab(a+b)= 0$. It can be easily seen that (\ref{FBCT2f2}) holds for all $x\in  \mathbb{F}_{2^n}$, which gives
$$
 {\rm FBCT}_{F}(a,b)=2^n.
$$

{\textbf{Case 2:}} $ab(a+b)\ne 0$. Let $c=\frac{a}{b}$ and $y=\frac{x}{b}$, we have $c\ne 0,1$. Then, (\ref{FBCT2f2}) is equivalent to
$$
b^{d_3}\big(y^{d_3}+(y+c)^{d_3}+(y+1)^{d_3}+(y+c+1)^{d_3}\big)=0.
$$
Since $b\ne 0$, we only need to consider the solutions of
\begin{equation}\label{FBCT2f2eq}
y^{d_3}+(y+c)^{d_3}+(y+1)^{d_3}+(y+c+1)^{d_3}=0.
\end{equation}
Since \(F(x)\) is locally-APN, according to Theorem \ref{niho theo}, for \(b_1 \neq 1\) the equation
\begin{equation}\label{41}
F(y+1) + F(y) = b_1
\end{equation}
has at most two solutions. Now suppose that \( y_1 \notin \mathbb{F}_{2^m} \) satisfies (\ref{41}) and also
\[
F(y_1)+F(y_1+1)+F(y_1+c)+F(y_1+c+1)=0,
\]
with \( c \neq 0,1 \). Then,
\[
F(y_1+c)+F(y_1+c+1)=b_1 .
\]
Thus \( y_1+c \) and \( y_1+c+1 \) are also two solutions of (\ref{41}).
Because (\ref{41}) has at most two solutions, we must have either
\( y_1+c = y_1 \) or \( y_1+c = y_1+1 \), i.e. \( c = 0 \) or \( c = 1 \), which contradicts the assumption \( c \neq 0,1 \). Hence, for any fixed \( c \neq 0,1 \), if (\ref{FBCT2f2eq}) has a solution, that solution must belong to \( \mathbb{F}_{2^m} \).   Next we consider solutions of (\ref{FBCT2f2eq}) in \(\mathbb{F}_{2^m}\). Since \(y^{d_3}=y\) holds for every \(y \in \mathbb{F}_{2^m}\), it follows that when \(c \in \mathbb{F}_{2^m}^{\sharp}\), every \(y \in \mathbb{F}_{2^m}\) is a solution.
If \(c \notin \mathbb{F}_{2^m}\), suppose some \(y \in \mathbb{F}_{2^m}\) satisfies (\ref{FBCT2f2eq}). Then \(y+c\) would also be a solution, but \(y+c \notin \mathbb{F}_{2^m}\), contradicting the fact that all solutions must belong to \(\mathbb{F}_{2^m}\). Therefore, for \(c \notin \mathbb{F}_{2^m}\), (\ref{FBCT2f2eq}) has no solution in \(\mathbb{F}_{2^m}\). This completes the proof.
\end{proof}

\begin{theorem}\rm\label{th3}
Let \(p\) be an odd prime. For any $(a,b)\in\mathbb{F}_{p^n}^2$, the value distribution of the multi-set $\left\{\nabla_F(a,b): a,b \in \mathbb{F}_{p^n} \right\}$ is shown in Table \ref{fbct221}.

\begin{table}[H]
\centering
 \caption{Second-order zero differential spectra}
\label{fbct221}
\begin{tabular}{|c|c|c|}
 \hline
  $\nabla_F(a,b)$  &  $(a,b)$  &{\rm  Multiplicity} \\\hline
    $p^n$ &  $ab=0$ &  $2\cdot p^n-1$ \\ \hline
   $p^m$ & $ab\neq0$, $\frac{a}{b}\in \mathbb{F}_{p^m}^{*}$
     &  $(p^m-1)(p^n-1)$ \\ \hline
     $1$  &  {\rm otherwise} &  $(p^n-p^m)(p^n-1)$ \\\hline
\end{tabular}
\end{table}
\end{theorem}

\begin{proof}
We mainly study the number of solutions $x\in\mathbb{F}_{p^n}$ to
\begin{equation}\label{FBCTf3}
  (x+a+b)^{d_3}-(x+a)^{d_3}-(x+b)^{d_3}+x^{d_3}=0
\end{equation}
for any $(a,b)\in \mathbb{F}_{p^n}^2$.

{\textbf{Case 1:}} $ab=0$. It can be easily seen that (\ref{FBCTf3}) holds for all $x\in  \mathbb{F}_{p^n}$, which gives
$$
 {\rm FBCT}_{F}(a,b)=p^n.
$$

{\textbf{Case 2:}} $ab\ne 0$. Let $c=\frac{a}{b}$ and $y=\frac{x}{b}$, where $c\ne 0$. Then, (\ref{FBCTf3}) is equivalent to
$$
b^{d_3}\big(y^{d_3}-(y+c)^{d_3}-(y+1)^{d_3}+(y+c+1)^{d_3}\big)=0.
$$
Since $b\ne 0$, we only need to consider the solutions of
\begin{equation}\label{FBCTf3eq}
y^{d_3}-(y+c)^{d_3}-(y+1)^{d_3}+(y+c+1)^{d_3}=0.
\end{equation}
When \(b_1 \neq 1\), the equation
\begin{equation}\label{44}
F(y+1) - F(y) = b_1
\end{equation}
still has at most two solutions.
Moreover, if \(y\) satisfies (\ref{44}), then \(-y-1\) also satisfies  (\ref{44}).
Suppose that \(y_1 \notin \mathbb{F}_{p^m}\) satisfies  (\ref{44}) and also
\[
F(y_1)-F(y_1+1)-F(y_1+c)+F(y_1+c+1)=0,
\]
with \(c \neq 0\). Then
\[F(y_1+c+1)-F(y_1+c)=b_1,\]
so \(y_1+c\) is another solution of  (\ref{44}).
Since  (\ref{44}) has at most two solutions and \(\{y_1,\,-y_1-1\}\) already gives two distinct solutions, we must have \(y_1+c \in \{y_1,\,-y_1-1\}\).
The case \(y_1+c = y_1\) gives \(c=0\), which is excluded. Therefore
\[
y_1 = \frac{-c-1}{2}.
\]
A direct substitution shows that \(y = \frac{-c-1}{2}\) indeed satisfies  (\ref{FBCTf3eq}).
Hence, for any fixed \(c \neq 0\), if  (\ref{FBCTf3eq}) has a solution, that solution must belong to
\[\mathbb{F}_{p^m} \cup \left\{\frac{-c-1}{2}\right\}.\]
Next we examine the solutions of (\ref{FBCTf3eq}) in \(\mathbb{F}_{p^m} \cup \{\frac{-c-1}{2}\}\). Because \(y^{d_3}=y\) holds for all \(y \in \mathbb{F}_{p^m}\), we see that if \(c \in \mathbb{F}_{p^m}^{*}\), every \(y \in \mathbb{F}_{p^m}\) is a solution. If \(c \notin \mathbb{F}_{p^m}\) and \(y \in \mathbb{F}_{p^m} \cup \{\frac{-c-1}{2}\}\) satisfies (\ref{FBCTf3eq}), then \(-y-c-1\) must also belong to \(\mathbb{F}_{p^m} \cup \{\frac{-c-1}{2}\}\). This forces \(y = \frac{-c-1}{2}\). Therefore, when \(c \notin \mathbb{F}_{p^m}\), (\ref{FBCT2f2eq}) has exactly one solution. This completes the proof.
\end{proof}

\begin{remark}\rm
Based on Theorems \ref{th11} and \ref{th3}, the FBCT and second-order zero differential spectra of \(F_1\) and \(F_2\) in Remark \ref{F1F2} can be directly obtained. Their value distributions are presented in Tables \ref{fbct22} and \ref{fbct221}, respectively. To the best of our knowledge, only six classes of power functions are known to have explicit values for all entries in their FBCT (see Table \ref{tablefb122}).

\begin{table}[H]
\caption{$F(x)=x^d$ over $\mathbb{F}_{2^n}$ with fully determined FBCT}
\label{tablefb122}
\centering
\begin{tabular}{|c|c|c|c|}
\hline
    $d$       & Conditions  & $\beta_c(F)$ & Ref.   \\
    [0.5ex]
    \hline
     $2^n-2$ & $n$ even & 4 &   \cite{edd.cc} \\  \hline
    $2^k+1$  &  $(n,k)=d$, $d\neq1$ & $2^n$ & \cite{edd.cc} \\   \hline
    $2^{2k}+2^k+1$  &  $n=4k$  & $2^{2k}$ & \cite{edd.cc} \\ \hline
     $2^{m+1}-1$  &  $n=2m$  & $2^{m}$ & \cite{depth} \\ \hline
       $2^{n-2}-1$  &  $3\nmid n$ or $3| n$ & 4 or 8 & \cite{me.24} \\   \hline
            $s(2^{m}-1)+1$  &  \makecell{$s=\frac{2^r}{2^r\pm1}$, $(r, m) = 1$, \\$n=2m$} & $2^m$ & This paper \\ \hline
       \end{tabular}
\end{table}
\end{remark}

\section{Conclusion}\label{conclusion}

In this paper, we mainly investigated Niho type power functions over  $\mathbb{F}_{p^{2m}}$. We determined the differential uniformity of all Niho type power functions. We showed that a  Niho type power function is locally-APN if and only if its Walsh spectrum is precisely four-valued, exactly taking the values in $\{ -p^m, 0, p^m, 2p^m \}$. Equivalently, the associated cyclic codes have four nonzero weights: $p^{m-1}(p-1)(p^m + k)$ for $k = 0, 1, -1, -2$. Furthermore, for Niho type locally-APN power functions, we also determined their differential spectrum, Walsh spectrum, associated cyclic codes, Feistel Boomerang Connectivity Table, and second-order zero differential spectra. The main contribution are summarized as follows.
\begin{itemize}
  \item For Niho type power functions that are locally-APN functions, we gave two necessary and sufficient conditions (see Theorem \ref{nihowalsh}).
If all four-valued cross-correlation Niho type power functions have been found, all locally-APN Niho type power functions can be obtained by according to our conclusion.
  \item For a class of special Niho type power functions, we determined its differential uniformity. Furthermore, we completely characterized its differential spectrum when these functions are locally-APN in Theorem \ref{niho theo}.
  \item We studied some properties of Niho type locally-APN power functions and determined the FBCT and the second-order zero differential spectra of a class of Niho type locally-APN power functions in Theorems \ref{th11} and \ref{th3}.
\end{itemize}

\end{document}